\documentclass[10pt,journal]{IEEEtran}
\IEEEoverridecommandlockouts
\usepackage{cite}
\usepackage[letterpaper, left=0.667in, right=0.667in, bottom=0.597in, top=0.792in]{geometry}
\usepackage{amsmath,amssymb,amsfonts,mathptmx,subfigure}
\usepackage{graphicx}
\usepackage{textcomp}

\usepackage{url,cite}
\usepackage{multirow}
\usepackage{mathtools}
\usepackage{hyperref}
\usepackage{comment}
\usepackage{xcolor}
\usepackage{amsmath,graphicx,latexsym,amssymb,amsthm}
\newcommand{\mR}{\mathbb{R}}

\newcommand{\mc}{\mathcal}

\newcommand{\tha}{\theta}

\renewcommand{\leq}{\leqslant}
\renewcommand{\geq}{\geqslant}

\newtheorem{prop}{Proposition}
\newtheorem{Def}{Definition}
\newtheorem{assume}{Assumption}

\usepackage{cite}
\usepackage{algpseudocode}
\usepackage{algorithmicx}
\usepackage{algorithm}
\usepackage{caption}
\usepackage{eufrak}
\mathtoolsset{showonlyrefs}

\DeclareMathOperator{\supp}{supp}

\newcommand{\kets}{|s\rangle}

\newcommand{\ket}{|\varphi\rangle}

\newcommand{\Ket}{|\Phi\rangle}
\newcommand{\Bra}{\langle \Phi |}

\newcommand{\brask}{\langle s\varphi_k|}

\newcommand{\ketsk}{|s\varphi_k\rangle}

\newcommand{\ketpsi}{|\psi\rangle}

\newcommand{\braskp}{\langle s\varphi_{k'}|}

\newcommand{\ketskp}{|s\varphi_{k'}\rangle}

\newcommand{\tr}{\text{Tr}}

\begin{document}

\title{Game of travesty: decoy-based psychological cyber deception for proactive human agents}
\author{Yinan Hu, {\it Student Member} and Quanyan Zhu, {\it Senior Member}}
\date{\today}
\maketitle
\begin{abstract}
\noindent The concept of cyber deception has been receiving emerging attention. The development of cyber defensive deception techniques requires interdisciplinary work, among which cognitive science plays an important role. In this work, we adopt a signaling game framework between a defender and a human agent to develop a cyber defensive deception protocol that takes advantage of the cognitive biases of human decision-making using quantum decision theory to combat insider attacks (IA). The defender deceives an inside human attacker by luring him to access decoy sensors via generators producing perceptions of classical signals to manipulate the human attacker's psychological state of mind. Our results reveal that even without changing the classical traffic data, strategically designed generators can result in a worse performance for defending against insider attackers in identifying decoys than the ones in the deceptive scheme without generators, which generate random information based on input signals. The proposed framework leads to fundamental theories in designing more effective signaling schemes. 
\end{abstract}
\setcounter{page}{1}

\section{Introduction}
Cyber deception has been a growing class of proactive defense techniques over the past several decades that contribute to combat increasingly intelligent, stealthy, and sophisticated attacks. Important cyber deception technologies including moving target defense \cite{jajodia2011moving_target_defense}, honey-x \cite{spitzner2003honeynet} (such as honeypots, honeytokens), etc help defenders reach a better security outcome against ever-growingly sophisticated attacks and threats, among which advanced persistent threats (APT) and insider threats \cite{collins2016insider_threat_guide} serve as two typical examples.  Reports have revealed that cyber deception technologies have reduced the cost arising from data breaches by $51\%$ in 2022 \cite{cyber_deception_report2022}.  
Cyber deception techniques take advantage of the human aspects to achieve two-fold purposes: one is to conceal the truth, and the other is to reveal the false.  The ultimate goal of applying defensive cyber deception techniques is to delay, stop, or interrupt attacks. Many techniques can achieve the concept of deception: dazzling, mimicking \cite{murphy2010application_deception_mimicking}, inventing, decoying \cite{carroll2011game_deception_decoy}. 
Useful defensive deception protocols characterize the strategic interactions among three classes of agents: defenders, users, and adversaries. A useful framework to design cyber deception mechanisms needs to capture  several main features. First, the defender must strategically treat users and adversaries with different purposes. In general, the defender should enhance the efficacy of access for a normal user and reduce the efficacy of access for adversaries. In addition, a sophisticated adversary behaves intelligently but also suffers from limitations arising from human aspects. 

Inter-disciplinary work is needed to help to develop next-generation deception techniques incorporating psychological models to characterize the behaviors of human attackers and system users. The interdisciplinary nature of the concept of deception constitutes a major challenge for researchers in building cyber deceptive defense systems.

Many game-theoretical models \cite{manshaei2013game_network_security} characterize the methods and mechanisms in the concept of deception in detection frameworks in cyber security. One major limitation of applying game-theoretical formulation solely to formulate threats is that such models often assume all agents are fully rational, while in real practices the behaviors of attackers and defenders often deviate from rationality \cite{manshaei2013game_network_security}, in part because devices in the network are operated by humans. 

One aspect of making breakthroughs in research in deception techniques is to adopt more accurate models in cognition to form more accurate predictions of the human attacker's behaviors. Such a direction is called cyber-psychology.  
 
Studies have shown that human reveals bounded rationality in decision-making due to a variety of cognitive biases. As a result, biases have played a cornerstone component in a successful deception mechanism not only in cyber security but also in social sciences. 

There are some other phenomena in cognitive science such as order effect, disjunction effect, violation of the total law of probability, etc, that are missed by previous deception mechanisms. New models need to be raised to characterize those phenomena. Game-theoretical models \cite{manshaei2013game_network_security} assume that both the sensor and the receiver present full rationality and may lead to a more conservative strategy for the defensive systems that manipulate data incoming to the sensors.  One difference between human decision-making theories and general decision theories is that human suffers from cognitive bias of various kinds such as margin effect, order effect, etc, that incur human agents to arrive at choices leading to suboptimal outcomes.  

There is literature catching the cognitive biases of humans arising from risk preferences and applying alternative frameworks such as quantal response \cite{mckelvey1998quantal_response} or prospect theory\cite{tversky1992advances_prospect_theory}. In behavioral economics \cite{mankiw2020principles_economics}, human's bounded rationality is presented in a variety of ways. Recently, there are experimental studies \cite{ferguson2021decoy_psychology} where experts play attackers who aim to hack into systems to gather information and aim to avoid decoys, while the defense system adopts cyber deception and cyber psychology techniques to prevent real systems from being attacked and credentials from being stolen.
The goal of the experimental study is to verify/invalidate the two hypotheses: one, defensive cyber tools and psychological deception impede attackers who seek to penetrate computer systems and exfiltrate information; and two, defensive deception tools are effective even if an attacker is aware of their use.

Experimental data essentially show that both hypotheses are true. But there is a lack of theories explaining why they are true. Constructing theories characterizing human agents' behaviors that take advantage of bounded rationality is beneficial in understanding human behavior to counteract them.

The quantum decision theories \cite{busemeyer2012quantum_cognition} catch the bounded rationality arising from order effect, disjunct effect, and violation of the total law of probability.  We are not arguing that the human brain acts like a quantum computer in the physical sense. Instead, we argue that the quantum decision theory functions as a \textit{generative parsimonious generative black-box model} for human's decision-making processes that have been corroborated by experiments such as \cite{kvam2021temporal_oscillation_quantum_preference}.
In this paper, we consider a scenario where sensors generate manipulated data for receivers, who are human agents. We assume the sensors constitute as part of the defense systems and the human agents want to attack sensors. The defensive systems aim at deceiving the human agents to mislead them to attack the wrong nodes. Such a system is called human-sensor system in cybersecurity. 

 The purpose of this paper is to develop an appropriate framework decoying as a method of cyber deception to characterize the sensor's manipulation of the traffic data and the attacker's strategies for attacking the sensors.  The challenge is to consider the receivers are human agents who make decisions suffering from a variety of bounded rationality arising from cognitive biases such as marginal effect, order effect, violation of the total law of probability, etc.  


In this paper, we propose the `traversty game' (TG) framework as a signaling game framework based on quantum information for constructing a cyber defensive deception to bring forth a desirable security scenario where the defender interacts with human adversaries to reduce the efficacy of the attacks by taking advantage of bounded rationality in human's decision-making process. The defender, or the deceptive defensive system, has a private type that characterizes the state of the system. That is, what connects the network and the human agent could be a regular network sensor or a decoy. It is common knowledge that a normal sensor and a decoy produce traffic data whose message volumes obey different distributions. The defensive system contains a sensor and a generator. The sensor collects original data traffic from the network and distorts data traffic. The generator is a mechanism that produces verbal messages to manipulate the human agent's perception of the classical traffic data. The cyber deceptive defensive system associates classical (maybe distorted) traffic data with the manipulated perception on the traffic data to deliver to human agents composite signals, which are characterized as `prospect states' 
\cite{yukalov2011decision_prospect_entanglement} in quantum decision theory \cite{busemeyer2012quantum_cognition}. Upon receiving the prospect states, the human agent (receiver) formulates the true state of the defensive system (a normal sensor or a decoy) into a quantum hypothesis testing problem and designs optimal prospect operator-valued measurements to minimize his weighted risk. The human agent then decides whether to access the system or not. The goal of the human agent, no matter his type, is to access sensors and avoid decoys. We thus formulate the human agent's objective as the weighted Bayesian risk that depends on the misdetection error rate and false alarm error rate. After the generator is implemented, both the defensive system and the human agent update their intelligence of each other through the Bayes' rule. The optimal behavior of the defensive deceptive system is to guide the human agent to access the true sensor while preventing it from accessing the decoy if the type were normal and vice versa.  Correspondingly, the optimal behavior of the human agent is to access the system if he were to find out the defensive system was likely normal and vice versa. Furthermore, we adopt the concept of repeated games with incomplete information \cite{zamir_game_theory_cambridge}\cite{aumann1995repeated_game} to study the temporal evolution of the strategies of both the defender and the human agent when both parties gather more and more information about each other. 

We formulate the decision problem for the human agent and derive that under mild assumptions, the anticipated behavior of the human attacker resembles quantum likelihood ratio test (QLRT). In the meantime, we formulate the defender's problem of designing optimal mixed type-dependent prospect states as a mixed-integer programming problem.   
We characterize how defense systems could make up the weakness of attackers as human agents by taking advantage of the bounded rationality. In particular,we adopt the concept of prospect probabilities \cite{yukalov2011decision_prospect_entanglement}, where the likelihood consists of two terms: utility factor and attraction factor \cite{sornette2020quantum_propsect_theory}. The utility factor represents probability that arises from the classical signals,  while the latter term does not arise from the actual data traffic but the perception of the data traffic due to quantum interference of different psychological state corresponding to the same classical signal.  The attraction factor could  lead the human agent towards (or away from) a certain choice. 

The main contribution of this work is two-fold. First, we develop a holistic framework to capture cyber-psychology techniques, specifying how a defender could implement cyber deception techniques by manipulating perceptions of signals to mislead an inside human attacker using her bounded rationality. Second, we illustrate and analyze human attacker's detection performance of decoys to show how strategically designed perceptions can influence human's decision-making and thus mitigate insider's attacks. Our analytical and numerical results provide hope for building next-generation deception mechanisms to combat human-related insider attacks in network security.

The rest of the paper is organized as follows. In section \ref{sec:formulation} we formulate the human-sensor system in cyber deception as a signaling game. In section \ref{sec:PBNE} we characterize the optimal behavior of the human agent and the cyber defensive system using the concept of equilibrium. 
In section \ref{sec: dynamic_scenario} we extend our signaling game formulation into a dynamic scenario, studying how the efficacy of the attacks evolve through time and how the defensive system and human agent can change their strategies as they both gather more intelligence from each other. 
In section \ref{sec:case_study} we provide a numerical case study on honeypots to illustrate our proposed framework. Finally we conclude in section \ref{sec:conclusion}. 

\subsection{Related work:}

   \paragraph{Game theory for cyber deception} In network security, game-theoretic frameworks have been widely applied for building proactive defense, particularly defensive cyber deception \cite{pawlick2021game_cyber_deception} to enhance the security and privacy for network users. 
Games with incomplete information \cite{zamir_game_theory_cambridge} provide a generic protocol to characterize the asymmetry of information induced by deception. Typical game frameworks that have been applied in network security include zero-sum games \cite{manshaei2013game_network_security},  Bayesian Stackelberg security games\cite{tambe2013improving_human_adversary}, partially observable stochastic games \cite{horak2017posg_security}. These complete or incomplete information game frameworks capture and interpret the attackers' and defenders' behaviors by computing appropriate concepts of equilibrium, depending on the information structure.   
In this work, we adopt the framework of a signaling game to characterize the relationship between the defensive deception system and human attacker, yet introduce quantum decision theory and the concept of quantum information to exploit the cognitive bias of the human attackers. 

\paragraph{Cyber deception through manipulating psychological states} There has been surging studies in formulating cyber deception techniques via psychological manipulation. 
Authors in \cite{ferguson2021decoy_psychology} have experimentally verified that it is not only the messages from the defensive system but also the perception of the messages, that will influence the human attacker's behavior. Authors in \cite{cranford2021cognitive_security_cyber_deception} propose an instance-based-learning (IBL) model of a human insider attacker using adaptive-control-of-thought-rational (ACT-R)\cite{anderson2014adaptive_control_rational_thought} as the cognitive architecture, which takes into consideration features in cognitive science: forgetting, power law of practice, partial matching process, etc. The IBL model also formulates how memory retrieval dynamics lead to cognitive biases. 
Our proposed framework adopts quantum decision theory, a generative parsimonious model to capture other biases in human's decision-making process \cite{busemeyer2012quantum_cognition} such as order effect and disjunction fallacy, etc. In addition, our proposed work focuses on how the defender system takes advantage of human attacker's biases by designing strategic generators in decoy systems that take advantage of human attacker's biases to combat. 

\paragraph{Insider threat/attack mitigation designs}
Several works have proposed guidelines for adopting honeypots in insider threats mitigation programs \cite{moore2015effective_insider_threat_mitigation}\cite{spitzner2003honeypots_insider_threat}.  
Game-theoretical frameworks have been adopted for formulating insider threats.
Authors in \cite{kantzavelou2010game_detection} use game-theoretical frameworks to develop detection mechanisms for insider-threats. Authors in \cite{joshi2020insider_threat_risk} adopt risk measures and extend risk analysis to cooperate with organizational culture. These works seek to contribute to an assessable understanding of the behaviors of adversarial insiders to develop more accurate best-responding strategies to combat insider threats but ignore the human aspects that lead to non-compliance of fully rational behaviors, such as the cognitive biases of various kinds in the human decision-making process. The authors in \cite{huang2021game_insider_threat_mechanism} have adopted the framework of mechanism design to address compliance and non-compliance for selfish and adversarial insiders. Our work adopts the concept of decoys to detect and monitor the behavior of insider attackers, deterring them from accessing normal sensors by taking advantage of cognitive biases of insiders to strategically design different perceptions of messages to influence their decision-making.

\subsection{Notations}
\label{sec:notation}
Throughout the paper, we use the following notations. We may introduce new notations in specific paragraphs later. 
We use the following notations:

$\mc{H}_C$: the (Hilbert) space overall signals;

$\kets\in \mc{H}_C$: a generic state associated with signal $s\in S$;

$\mc{H}_I$: the (Hilbert) space over all states of mind;

$\ket\in\mc{H}_I$: a generic state of mind;

 $\mc{H} = \mc{H}_C\otimes \mc{H}_I$: the Hilbert space (over the set of real numbers $\mR$) of all 'prospects'.
 
 $\mc{H}^*$: the dual space of $\mc{H}$;
 
$\mc{S}$: the subset of positive,Hermitian, bounded operators on $\mc{H}$ whose trace is $1$;  

$S$: the space of signals;

$\Delta(\cdot)$: the set of probability measures over the given space;

$\mathbf{1}$: the identity operator. Its domain and range depend on the context; 

$p_k\in\Delta(X)$: the common prior/common posterior belief of the true state after $k-1$ observations have been generated;  

$X=\{1,2\}$: the state space of the system. A generic state is denoted as $x$: $x=1,2$ represents the system is abnormal and normal, respectively; We denote $\text{dim}(X) = M = 2$. 

$a,b\in \mR^{S\times K}$: generic perception matrices from the defender based on its true type $0,1$.

In addition, for any operator $A\in B(\mc{H})$, we denote its conjugate transpose as $A^{\dagger}$.

\begin{figure}
    \centering
\includegraphics[width=0.9\linewidth]{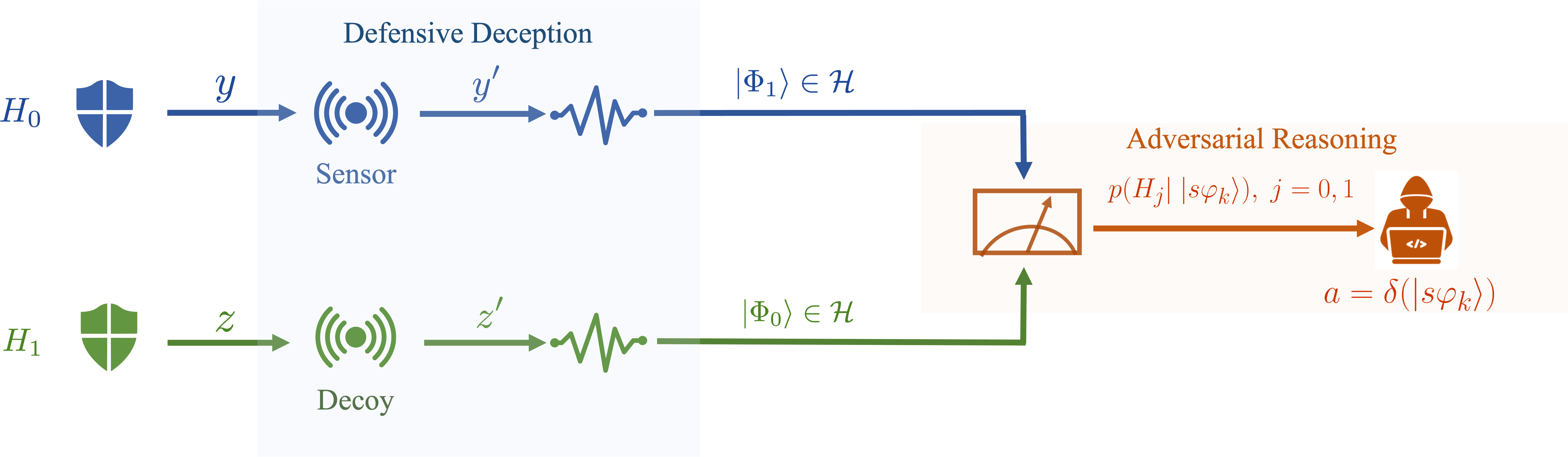}
    \caption{The human-sensor system in a network security scheme. The defensive cyber deception system consists of a normal sensor and a decoy, each of which is cascaded with a generator. The human agent is a receiver taking manipulated network traffic associated with perception messages. The normal sensor and the decoy produce manipulated traffic data obeying different distributions. The location of the decoy is a private type for the defensive system.  The receiver is also associated with three private types: user, prospect attacker, and quantum attacker. The goal of the defensive system lure the human attacker to access the decoy rather than the normal sensor. The goal of a human attacker aims at recognizing the decoy to avoid while making use of normal sensors.}
    \label{fig:human_sensor_system_b}
\end{figure}
\section{The Formulation of Traversty Game}
\label{sec:formulation}

\subsection{Purpose of formulation}

Insider threats \cite{collins2016insider_threat_guide} has long been an important issue in cyber-security because, in contrast to external attackers, insiders are aware of the structure of the defensive system, know about vulnerability, and more likely to launch strategic attacks to destroy the system effectively. Thus defensive deception techniques such as decoy systems have been implemented for the detection and mitigation of insider threats \cite{spitzner2003honeypots_insider_threat}. The goal of designing a defensive deception mechanism is to expose the vulnerabilities of their decoys to attract adversaries to access honeypots/decoys to trigger alerts and thus the defensive system can gather their information. 
To address the challenge, we need a novel configuration of decoy sensors and normal sensors to develop a next-generation defensive cyber deception system that exploits human biases to attract human attackers to focus on the decoy. Previous literature has pointed out \cite{manshaei2013game_network_security} that future defensive systems in network security must also consider human factors when predicting attackers' behaviors.  Human agents are subject to decision-making biases and exhibit other non-rational behavior. To this end, it is effective to introduce cyber-psychological techniques. Cyber-psychology \cite{mcalaney2018psychological_security} is the scientific field that integrates human behavior and decision-making into the cyber domain allowing us to understand, expect, and further influence attacker's behaviors. Experimental studies \cite{ferguson2021decoy_psychology} have shown that by providing information to the human factor, their mental model on the cyber defensive system was influenced and their decisions changed accordingly. To our best knowledge, there is still a lack of theoretical frameworks to interpret and address how those cyber-psychological methods can work effectively to mitigate attacks.    

\subsection{The game formulation}
In this section, we propose a game-theoretical framework on cyber defensive deception systems that mitigates insider threats by adopting cyber-psychological techniques. We show that cyber-psychological techniques can demonstrate a better deterrence of insider threats than their classical counterparts. 
We consider the protocols whose scheme is depicted in Figure \ref{fig:human_sensor_system_b}. In short, the defensive deception system (she) and the receiver (human agent, he) play a signaling game $\mc{G}$. The defensive system consists of two sensors: one normal, and one decoy, each of which is cascaded with a generator that generates psychological messages reflecting the perception of the manipulated data traffic. The defensive system connects one of the sensors to the human agent. The (human) receiver knows there is a decoy sensor but the placement of the decoy is unknown and serves as the defensive system's private type $x\in\{0,1\}$. He could only make decisions based on classical traffic data and the perception messages associated with the data.  Normal sensor accepts observations obeying distribution $g_0$, while decoy accepts observations obeying $g_1$. Denote $\hat{s}$ as random variables characterizing the random observations $s$ as corresponding realizations.   We say that 
human-agent faces a hypothesis testing problem:
    \begin{equation}
        H_1: \hat{s} \sim g_1(s),\;\;\;  H_0: \hat{s}\sim g_0(s).
    \end{equation}

The goal of the defense system is to strategically configure the normal sensor, the decoy sensor, as well as the generators to attract human attackers to access the decoy. The defensive system earns a reward when the adversarial agent accesses the decoy since such access provides information and triggers the alert of the defensive system \cite{ferguson2021decoy_psychology}. 

Meanwhile, the cyber deception system obtains observations $y$ from the network traffic (see figure \ref{fig:human_sensor_system_b}). Both normal and decoy sensors produce manipulated observations $s'$ and pass them into cascading generators. 
Based on manipulated observations $s'$ and the private type $x_0$, the connecting generator produces psychological signals characterized by a set of coefficients $\{a_{sk},b_{sk}\}_{s,k}$. The distorted observations together with psychological signals constitute the prospect state $|\Phi_1\rangle,|\Phi_0\rangle$ in the following way:
    \begin{equation}
        |\Phi_1(s)\rangle = \sum_{k}{a_{sk}\ketsk},\;|\Phi_0(s)\rangle = \sum_{k}{b_{sk}\ketsk}, 
        \label{prospect_state}
    \end{equation}
where we inherit Dirac's notations as introduced in section I.
Such a quantum state can be interpreted as messages announced to change the human agents' perceptions.  Such a generator produces stochastic messages manipulating the user's perception of messages. One example is an announcement like `the message comes from a real sensor'.  For instance, authors in \cite{ferguson2021decoy_psychology_deception} conducted experiments where the test takers are informed of the existence of a decoy system. The quantum-mechanical representation of messages can be referred to in \cite{busemeyer2012quantum_cognition}.        
     Upon observing $\Ket$, the human agent, randomly measures the prospect using one of the prospect basis $\ketsk$ and updates his prior belief on the defender's type. His mind is a composite prospect state \cite{yukalov2011decision_prospect_entanglement}. The human agent arrives at a decision $\alpha = \delta(\ketsk)\in [0,1]$, indicating the probability that the attacker thinks the hypothesis $H_1$ holds true.  

\paragraph{Sensor's (Defender's) problem}
The defender strategically designs manipulated classical observations from both sensors to mislead or promote human judgment. In the meantime, the defender creates type-dependent perceptions $a = (a_{sk})_{s,k}, b = (b_{sk})_{sk}\in\mR^{S\times K}$ regarding every signal $s\in S$ corresponding to the type $x=1$ and $x=0$ accordingly. The defender will earn a positive reward when the normal user accesses the normal sensor and a negative one when the attacker accesses the normal user or avoid accessing the decoy. 

\paragraph{Sensor's actions and strategies} Depending on the true type $x$ of deception system and well as the classical signal $s$, a generic defender's action involves a pair of prospect states $|\Phi_x(s)\rangle
\in\mc{H},\;\;$ We may also equivalently characterize the sensor's actions as two matrices $(a_{sk},b_{sk})$ since they can be written as in \eqref{prospect_state}.

If we consider that the defensive system adopts mixed strategies, we could characterize the mixed strategies as density operators $\rho_1,\rho_0$ as follows:
\begin{align}
         {\rho}_1 & = \sum_{s,k,k'}{f_1(s){a_{sk}a_{sk'}|s\varphi_k\rangle\langle s\varphi_k'|}},  
        \label{eq:rho1_prospect_state}
        \\
        {\rho}_0 & = \sum_{s,k,k'}{f_0(s){b_{sk}b_{sk'}|s\varphi_k\rangle\langle s\varphi_k'|}},  
       \label{eq:rho0_prospect_state}
    \end{align}
where $f_1,f_0$ are probability density functions over $M$. Another way to characterize the sensor's actions is via the utility factor and attraction factor. 
Denote 
\begin{equation}
\begin{aligned}
    \langle \Phi_1 |P |\Phi_1\rangle &= \sum_{s,k,k'}{a_{sk} a_{sk'}\brask P \ketskp} \\
    & = \sum_{s,k}{a^2_{sk} \brask P\ketsk} + \sum_{s,k\neq k'}{a_{sk}a_{sk'}\brask P \ketskp}\\
    & \equiv u_1(s) + q_1(s) \\
     \langle \Phi_0 |P |\Phi_0\rangle &= \sum_{s,k,k'}{b_{sk} b_{sk'}\brask P \ketskp} \\
    & = \sum_{s,k}{b^2_{sk} \brask P\ketsk} + \sum_{s,k\neq k'}{b_{sk}b_{sk'}\brask P \ketskp}\\
    & \equiv u_0(s) + q_0(s), \\
\end{aligned}
\end{equation}
where $u$ is the utility factor and $q$ is the attraction factor of the prospect state upon the decision operator $P$. We here define 
\begin{align}
    u_1(s) &= \sum_{\ketsk \in \mc{R}}{{a^2_{sk}}}, \label{util_1} \\
     q_1(s) &= \sum_{\ketsk,\ketskp \in \mc{R}}{a_{sk}a_{sk'}},
     \label{attr_1}\\ 
      u_0(s) &= \sum_{\ketsk \in \mc{R}}{{b^2_{sk}}},
      \label{util_0}\\ 
       q_0(s) &= \sum_{\ketsk,\ketskp \in \mc{R}}{{b_{sk}b_{sk'}}} 
       \label{attr_0}
\end{align}

According to \cite{vincent2016calibration_qdt}, we adopt some calibration rules to construct the attraction factor $p$ so that it is related to the utility factor $u$ as 
\begin{equation}
    q_j(s) =\zeta \min\{u_j(s),1-u_j(s)\},\; j =0,1,
\end{equation}
where we can further denote $\zeta\in[-1,1],\;$ and notice again that $u_j(s)\in[0,1]$. 
Furthermore $u_j(s) = 1$ only when all $\ketsk \in \mc{R}$ for all $k\in K$ for the given $s$. The opposite goes with $u_j(s) = 0$. A similar goes with $u_0(s) = 1$. Here we use the parameter $\zeta$ to simplify the hyperbolic tangent function used in \cite{sornette2020quantum_propsect_theory}. We introduce the following assumption:
\begin{assume}
\label{assume:existence_a_b_coefficient}
    The coefficients $a, b\in \mR^{S\times K}$ as in \eqref{prospect_state} exist for every $u_1,u_0\in [0,1]^{S}$.
 \end{assume}
 
Assumption \ref{assume:existence_a_b_coefficient} guarantees that we can construct $a_{sk},b_{sk}$ using these equalities (of course, there may not only be one exact choice of $a,b$ reaching the same utility factor and attraction factor). Now it is equivalent to use the quantities defined in \eqref{util_1}\eqref{attr_1}\eqref{util_0}\eqref{attr_0} to characterize the defender system's behavior.

\paragraph{Defender's utility/loss function}
We now formulate the defender's loss function to minimize. The goal of the defender is to mitigate the human attacker's performance in identifying decoys so his objective function is the genuine detection rate introduced in \eqref{PD_PF_quantum}. 
 This is because every time the human attacker commits an error, or equivalently, access to the decoy sensor, an alert will be triggered and the defensive system can gather intelligent information from the human agent \cite{pawlick2021game_cyber_deception}.

The defensive deception system designs type-dependent distributions $\rho_1,\rho_0$ under the type $x$ by minimizing the following objective $J^x_S: X\times \Delta(M)\times \mc{B}(\mc{H}) \rightarrow \mR$ as 
\begin{equation}
\begin{aligned}
    J^x_S(x,\rho_x, P_1^*) &= \tr(\rho_1 P^*_1),
\end{aligned}
\end{equation}
where $P^*_1\in {B}(\mc{H})$ denotes the optimal prospect-projection-based decision policy for the human agent. Using the theory of potential games \cite{lloyd_1996potential_game}, we know the human is equivalent to minimize the following objective $J_D: \Delta({M})\times \mc{B}(\mc{H})\rightarrow \mR$:
\begin{equation}
\begin{aligned}
   \underset{\substack{a,b \\ f_1,f_0}}{\min}& \;J_D(a,b,P^*_1) = J^1_S(1,\rho_1,P^*_1) + J^0_S(0,\rho_0,P^*_1)  \\
   \Leftrightarrow \underset{\rho_1,\rho_0}{\min}&\; \sum_{\delta^*(\ketsk)>0}{\brask \rho_1 \ketsk} + 1,
   \label{problem:P1}
\end{aligned}
\end{equation}
where we compute the trace using the prospect basis $\{\ketsk\}_{s,k}$. If we adopt $u_1,u_0,f_1,f_0$ as the defender's type-dependent strategy, we can introduce the objective function $F: [0,1]\times [0,1]\times L^1(S)\times L^1(S) \rightarrow \mR$ as follows:
\begin{equation}
  \underset{\substack{f_1,f_0\\ u_1,u_0}}{\min }F(u_1,u_0,f_1,f_0) \Leftrightarrow   \underset{\substack{f_1,f_0\\ u_1,u_0}}{\min }\sum_{s\in\mc{R}_s}{f_1(s)u_1(s)}.
    \label{problem:P2}
\end{equation}
with $\mc{R}_s : =  \{s: \exists k,\; \ketsk\in\mc{R}\}$.
\begin{prop}
\label{prop:equiv_optimization_defender}
    Let $(a^*,b^*)$ be an optimal solution for the optimization problem \eqref{problem:P1}, Let $u^*_1,u^*_0: S\rightarrow [0,1], q^*_1,q^*_0: S\rightarrow [-1,1]$ be the optimal solution for the optimization problem  \eqref{problem:P2}. Then we can construct the relation in \eqref{util_1}\eqref{attr_1}\eqref{util_0}\eqref{attr_0}.
\end{prop}
The proof can be viewed in the appendix \ref{appd:equiv_optimization_defender}. 

\paragraph{The belief updates}
Upon receiving the prospect state $\Ket \in \mc{H}$, the human agent first updates the prior belief regarding the defender's type into posterior belief:
\begin{equation}
p(H_x|\;\ketsk) = \frac{p(H_x)\tr(P_{sk}\rho_x P^{\dagger}_{sk}) }{p(H_1)\tr(P_{sk}\rho_1 P^{\dagger}_{sk})  + p(H_0)\tr(P_{sk}\rho_0 P^{\dagger}_{sk}) },\;x=0,1,    
\label{posterior_belief}
\end{equation}
where $P_{sk}\in \mc{H}$ is the projection operator upon a specific the prospect state basis $\ketsk$: that is, $P_{sk} = \ketsk\brask$. 

\paragraph{Human's actions}
The human agent first can estimate the defender's strategies, characterized as mixed prospect states $\rho^*_1,\rho^*_0$ at equilibrium. Thus he can construct two density operators under each hypothesis as psychological prospects in \eqref{eq:rho1_prospect_state}\eqref{eq:rho0_prospect_state}.  The human's action $\alpha\in [0,1]$ characterizes the probability that the human agent thinks the traffic data come from a decoy (therefore not to access). The human agent arrives at a decision rule $\delta: \mc{H}\rightarrow [0,1],\; \alpha= \delta(\ketsk)$ upon receiving the prospect state $\ketsk\in \mc{H}$ from the deceptive defense system through a measurement operator $P\in B(\mc{H})$ as follows:
\begin{equation}
    \delta(\ketsk) = \brask P \ketsk. 
    \label{def: decision_measurement}
\end{equation}
Equivalently, the human agent's strategy space is the space of all projective operator-valued measurements (POVM). The human agent applies the concept of Neyman-Pearson hypothesis testing scenario \cite{neyman_pearson1933}: that is, a human agent aims at maximizing the probability of detection (accessing the normal user) while constraining the probability of false alarm (choosing to access while the target sensor is a decoy). Based on the $\ketsk$, the human attacker's empirical false alarm rate is $p(H_0|\;\ketsk)$ so we can express his strategy space $A_H$ as follows.
\begin{equation}
A_H = \{ \delta :\mc{H}\rightarrow [0,1]:\delta(\ketsk) p(H_0|\;\ketsk)<\beta\},
\label{A_H_human_attacker_strategy}
\end{equation}
where $\beta$ is the tolerance that the human agent could have regarding his false alarm rate. The posterior belief $p(H_j|\;\ketsk)$ is expressed in \eqref{posterior_belief}.
\paragraph{Human agent's type-dependent utility/loss function}
The human attacker wants to avoid decoys and access normal sensors. The human agent also suffers from cognitive biases characterized by quantum decision theory. Now the prior belief $p$ is constructed and updated in terms of the defense system's type $x$.   We now assume that the human agent arrives at a decision based on the posterior belief $p(H_1|\Phi)$: if it is too high, then the human agent will choose $0$ to avoid the cost of low.
Human's optimization problem can be expressed as 
\begin{equation}
    \begin{aligned}
        \underset{\delta  \in A_H}{\max} \;\delta(\ketsk) p(H_1|\;\ketsk)
        \label{A_H_human_attacker_optim_problem}
    \end{aligned}
\end{equation}

\subsection{Game elements}
We can now summarize our discussions in the previous section and propose our novel protocol for the game in the following definition.

\begin{Def}[`Traversty game'(TG)]
\label{Def:signaling_game}
 We define `game of travesty', a signaling game $$\mc{G} = \langle \mc{I}, X, A_S, A_{H}, F_S, J_H, p\rangle,$$
where $\mc{I} = \{\text{defender},\text{human attacker}\}$ represents the set of players; $x\in X=\{0,1\}$ be the defender's type (normal or decoy); $A_S = M\times H$ be the classical message space from the defender; $A_{H} \subset [0,1]$ represents the human agent's action space; $\mc{H}$ be the space of perceptual message from the generator; $F_S: [0,1]^2\times [L^1(S)]^2\times B(\mc{H})\rightarrow \mR $ be the defender's objective function; $J_H: M \times \mc{H}\times A_H\rightarrow \mR$ be the human agent's type-dependent objective function; $p\in\Delta(X)$ be the common prior belief of the private types of the defender and the human agent.
\end{Def}

\subsection{Relation to classical signaling games}

The proposed traversty game can be considered as a generalization of the hypothesis testing game raised in \cite{hu2022_evasion_game_CDC} with two-sided incomplete information, heterogeneous receivers, and adoption of quantum probabilistic model. The framework in \cite{hu2022_evasion_game_CDC} consolidates hypothesis testing formulation into signaling game framework \cite{crawford1982strategic_information_transmission} where one party, upon knowing the true hypothesis, can strategically manipulate observations to undermine the detection performance.  If the defender cannot design perceptions of classical messages using generators, then the travesty game framework reduces to hypothesis testing game framework in \cite{hu2022_evasion_game_CDC}. The adoption of the prospect state enhances the cyber deception design by taking advantage of human's bounded rationality to provide the defender extra degrees of freedom. Such degrees of freedom characterize how the human agents' `perceptions' of classical messages can contribute to their decision-making process. 

There are several scenarios where the defender's strategies are reduced to classical counterparts. Denote $a,b$ as the matrices of coefficients of the defender in \eqref{prospect_state}. Then when $a = R_1I, b = R_0 I$, where $R_0, R_1$ are some column permutation matrices and $I$ is an identity matrix, then the `quantum effect' vanishes as the defender associates a unique fundamental `mindset'` regarding every classical signal $s$. 



\subsection{Equilibrium Analysis}
\label{sec:PBNE}

We aim at computing the perfect Bayesian Nash equilibrium \cite{fudenberg1998game} (PBNE) to characterize the behaviors of the defense system and the human agents. We can define PBNE of the game $\mc{G}$ as follows:
\begin{Def}[Perfect Bayesian Nash Equilibrium for the game $\mc{G}$]
   We define the perfect Bayesian Nash equilibrium (PBNE) of the signaling game $\mc{G}$ as the following tuple $(u^*_1,u^*_0,\delta^*,p)$
   meeting the following requirements:
   \begin{enumerate}
       \item (Human agent's sequential rationality)
    \begin{equation}
        \delta^*({\ketsk}) \in\arg\underset{\delta\in A_H} {\min} J_{H}(\ketsk,u^*_1,u^*_0,\delta),
\label{def:human_attacker_seq}    
\end{equation}
\item (Defensive system's sequential rationality) 
\begin{equation}
    (u^*_1,u^*_0)\in \arg\underset{u_0,u_1}{\min}\;  F(u_1,u_0,f_1,f_0,\delta^*),\;x\in\{0,1\},     \label{def:pbne_defender_seq_rationality}
    \end{equation}
       \item (Belief consistency) 
       The belief is updated according to Bayes' rule:
       \begin{equation}
p(H_j|\; \ketsk) = \frac{p(H_j,|\;\ketsk)\brask\rho_j\ketsk}{\sum_{\substack{j'=0,1}}{p(H_{j'})\brask\rho^*_{j'}\ketsk }},\;\;j=0,1.
\end{equation}
   \end{enumerate}
\end{Def}

We can derive the human agent's optimal decision rule as follows.
\begin{prop}
\label{prop:human_agent_optimal}
Consider $\mc{G}$ to be the travesty game in definition \eqref{Def:signaling_game}. Let $(a^*_{sk},b^*_{sk})_{s\in S, k\in K}$ be defender's coefficients of optimal type-dependent prospect states satisfying the utility factors $u^*_1,u^*_0$ in \eqref{util_1}\eqref{util_0}, which are characterized as the defender's strategies at equilibrium \eqref{def:pbne_defender_seq_rationality}. Then the human attacker's optimal decision rule $\delta^*: \mc{H}\rightarrow [0,1]$ at equilibrium defined in \eqref{def:human_attacker_seq} receiving the prospect state $\ketsk$ reduced from superposition state $\Ket$ can be derived as
\begin{equation}
  \delta^*(\ketsk) =  \begin{cases}
        1 &  \frac{f_1(s) (a^*_{sk})^2}{f_0(s) (b^*_{sk})^2 }> (\frac{1}{\beta}-1)\frac{p(H_0)}{p(H_1)}, \\
        0 &  \mbox{otherwise}
    \end{cases}
    \label{def:diagonal}
\end{equation}
\end{prop}
\begin{proof}
    See the appendix.
\end{proof}

The optimal decision rule $\delta^*$ decomposes the space of prospect states $\mc{H}$ into region of rejection $\mc{R}$ and region of acceptance $\mc{R}_0$ as follows:
\begin{equation}
    \mc{R} = \text{span}\{\ketsk\}_{\delta(\ketsk)=1},\;\mc{R}^{\perp} = \text{span}\{\ketsk\}_{\delta(\ketsk)=0}.
    \label{def:rejection_region}
\end{equation}

Referring to the definition of the decision rule \eqref{def: decision_measurement} we notice that the diagonal elements of $P_1$ have been specified. We now assume the off-diagonal elements as 
\begin{equation}
    \braskp  P_1 \ketsk = \begin{cases}
        \frac{1}{N_s} & 
 \ketsk,\ketskp\in\mc{R}, \\
        0  & \mbox{otherwise},
    \end{cases}
    \label{projection_off_diagonal}
\end{equation}
where $N_s$ is the number of vectors among $\{\ketsk\}$ that lie in $\mc{R}$.

\begin{prop}
    The operator $P_1$ defined in \eqref{projection_off_diagonal} and \eqref{def:diagonal} is a projection operator.
\end{prop}
\begin{proof}
   It is clear that $P_1\geq 0$ from proposition \ref{prop:human_agent_optimal}. From \eqref{projection_off_diagonal} we know $P_1$ is symmetric. In addition $P^2_1\ketsk = P_1(\sum_{\ketsk\in \mc{R} }{\frac{1}{N_s}\ketsk}) = P_1\ketsk$ so $P^2_1=P_1$. Thus $P_1$ is a projection operator \cite{reed1972methods_vol1}. 
\end{proof}

\begin{assume}[No change of classical message]
\label{assume:no_change_classical_message}
We assume that the defensive deception system does not change the classical message. That is, $g_1 = f_1,\;g_0 = f_0$.    
\end{assume}

Equipped with the human agent's optimal decision rule $\delta^*$ in \eqref{prop:human_agent_optimal}, we can simplify \eqref{def:pbne_defender_seq_rationality} and derive the following.

\begin{prop}
    \label{prop:generator_equil_strategy}
Let assumption \ref{assume:no_change_classical_message} hold. Let $\mc{G}$ be the signaling game in Definition \ref{Def:signaling_game}. Let $\delta^*$ be the human attacker's optimal decision rule defined in \eqref{def:human_attacker_seq} upon receiving prospect states with coefficients $a^*,b^*$ defined in \eqref{prospect_state}.  Denoting $\tau_s = \frac{p(H_1)f_0(s)}{p(H_0)f_1(s)}(\frac{1}{\beta}-1)$, we thus derive the defender's strategies $u^*_1(s),u^*_0(s)$ at equilibrium defined in \eqref{def:pbne_defender_seq_rationality} as by the following cases for every $s\in S$:
\begin{enumerate}
   \item When $\tau_s>1$, we pick $u^*_1(s) = 0$ and thus $u^*_0(s) = 0$; 
    \item When $0<\tau_s<1$, we pick region of acceptance until 
\begin{equation}
    1 - u_1(s) = \tau_s.
\end{equation}
so $1 - u^*_0(s) = 1$ or equivalently $u^*_0(s)=0$. Then $u^*_1(s) = 1-\tau_s$.
\end{enumerate}
The corresponding region of classical rejection can be written as $\mc{R}_s = \{s:\;0<\tau_s<1\}$.
\end{prop}
\begin{proof}
    The proof is provided in the appendix \ref{appd:sol_signaling_lam_0}. 
\end{proof}

After obtaining $u^*_1(s), u^*_0(s),\;s\in S$, we can reconstruct the optimal prospect states $a^*,b^*$ by solving \eqref{attr_0}\eqref{attr_1}\eqref{util_0}\eqref{util_1}.
To measure the efficacy of cyber deception systems in counteracting the attacks, we can define the genuine detection rate and false alarm rate 
\begin{equation}
    P_D(\tau) = \tr(\rho^*_1  P^*_1(\tau)),\; P_F(\tau) = \tr(\rho^*_0 P^*_1(\tau)).
    \label{PD_PF_quantum}
\end{equation}
As a comparison, we denote the vanilla detection rate and false alarm rate of the insider attack (IA) as 
\begin{equation}
    \bar{P}_D(\tau) = \sum_{s:\delta^*(s;\tau)=1}{f_1(s)},\; \bar{P}_F(\tau) = \sum_{s:\delta^*(s;\tau)=1}{f_0(s)}.
        \label{PD_PF_classical}
\end{equation}

We now show that the role of the generator is to create more room for the attacker to deceive the human attacker to lower their probability of identifying the decoy system.

\textbf{Remark:} We also find out when $\tau\rightarrow \infty$ (the whole region of $S$ is of classical acceptance region ) or $\tau\rightarrow 0$(the whole region of $S$ is of classical rejection region), the detection rate $P_D(\tau)$ is close to $\bar{P}_D(\tau)$. That is, the quantum effect in decision-making vanishes when the prospect probability is close to $1$ or $0$, which is consistent with the discussion in Vincent's work \cite{vincent2016calibration_qdt} in quantum prospect theory.

\subsection{Some metrics evaluating the quantum advantage/disadvantage}
\paragraph{Quantum advantage and quantum disadvantage}
We can define the following metrics as quantum advantages/disadvantages as follows. 

\begin{Def}[Quantum advantage/disadvantage]
   We define the quantum advantage of 
   \begin{equation}
       QA(\tau) = P_D(\tau)/\bar{P}_D(\tau), 
   \end{equation}
   where $P_D:\mR\rightarrow \mR$ is the detection rate for the human attacker under manipulation defined in \eqref{PD_PF_quantum} and  $\bar{P}_D:\mR\rightarrow \mR$ be the counterpart for a non-adversarial human attacker without bounded rationality defined in \eqref{PD_PF_classical}.
   \label{Def:Quantum_advantage1}
\end{Def}
The quantum advantage (QA) is a crucial evaluation for the effect of introducing the generator in the defender system. It depends on the threshold $\tau$ as well as the calibration parameter $\zeta$. It measures the impact of manipulation of mind states upon human attacker's performance on detecting decoys.  We say that the human attacker gains a quantum advantage in identifying decoys if $QA(\tau)>1$ and suffers a quantum disadvantage if $QA(\tau)<1$. 

\begin{prop}
\label{prop:PD_PD_bar}
Let $QA$ be the quantum advantage in definition \ref{Def:Quantum_advantage1}.    Then for all choices of $\tau>0$ and for all choices of $f_1,f_0\in L^1(S)$, we arrive $P_D(\tau)\leq \bar{P}_D(\tau)$.
 \begin{equation}
        0\leq QA(\tau)\leq 1+\zeta.
    \end{equation}
\end{prop}
\begin{proof}
    See the appendix \ref{appd:PD_PD_bar}
\end{proof}

\section{Dynamic scenario}
\label{sec: dynamic_scenario}
We now extend $\mc{G}$ into a multi-stage game $\mc{G}^N$ with finite horizon $N$. For each stage $k\in [N]$, the sensor system generates manipulated observations and the human agent launches access to one of the sensors. After both the defense system and the human agent take actions, cost/reward is incurred.  The system belief on the defender's true type is updated.  We assume that the  defender never changes his type during the game. Therefore, the defender system exposes more about his type (normal sensor or decoy) as she produces more messages. We introduce the concept of the history of the actions taken by both the sensor and the human agent as follows.


\begin{Def}[History of action profiles]
We define the history of action profiles up to stage $N$, denoted as $h^{(j)}\in \mc{H}^{\otimes j}\times [0,1]^{\otimes j},\;j\in[N_1]$, as follows:
\begin{equation}
    h^{(j)} = (\psi^{(j)}, \alpha^{(j)}),
\end{equation}
where $\ketpsi^{(j)} = (\ketpsi_1,\dots,\ketpsi_{j})\in \mc{H}^{\otimes j}$ as a generic history of base vectors from the prospect state up to stage $j$ and  $\delta^{(j)}(\psi^{(j)}) \in [0,1]^{\otimes j} = A^{\otimes j}_H$ refers to history of the detector's actions up to stage $j$. 
\end{Def}

In general, at the beginning of every stage $j$, the defender's mixed strategy and the human agent's optimal decision rule should depend on the history $h^{(j-1)}$. Here we denote $\psi\in \{\ketsk\}_{s,k}$ as a generic base vector in the prospect state basis. We assume in the following that 
\begin{assume}[Action-independent assumption]
At every stage $j\in[N_1]$, the human attacker's optimal decision rule $\delta^{(j)*}(\cdot|h^{(j)})\in \bar{\Gamma}^{(j)}$, the attacker's optimal mixed strategies of generating manipulated messages $u^{(j)*}_1(\cdot|h^{(j)}),u^{(j)*}_0(\cdot|h^{(j)})\in [0,1]^S$ and the posterior belief $p(\cdot|h^{(j)})$ depend only on the attacker's history of mixed strategies. Specifically, we have for $k\in\{0,1\}$,
\begin{align}
    {\bar{\delta}}^{(j)*}(\psi_j|h^{(j)}) &= {\bar{\delta}}^{(j)*}(\psi_j|\psi^{(j)}),
    \label{action_hisotry_dependent}\\
     u^{(j)*}_k(\cdot|\;h^{(j)}) &=   u^{(j)*}_k(\cdot|\;\psi^{(j)}),
     \label{msg_history_dependent}\\
    p(H_k|h^{(j)}) &=  p(H_k|\;\psi^{(j)}).\;
    \label{belief_system}
\end{align}
\end{assume}
The three assumptions \eqref{action_hisotry_dependent}\eqref{msg_history_dependent}\eqref{belief_system} imply the only useful information accumulated throughout stages is the attacker's mixed strategies. The multi-stage game $\mc{G}^{N_1}$ based on the base game $\mc{G}$ is played as follows: before stage 1, the defender observes from Nature his type ($H_0$ or $H_1$); at the beginning of stage $j\in [N_1]$, the sensor observes the sender's message and sends prospect state $\Ket \in \mc{H}$ according on his mixed strategies $\rho^{(j)}_1,\rho^{(j)}_0\in B(\mc{H})$ to the human agent, who makes a decision based on the current prospect state and the  history of prospect states $\psi^{(j)}\in \mc{H}^{\otimes j}$ regarding the defender's type. 
At stage $j\in [N]$, we are now ready to define the human agent's hypothesis testing game problem for the human agent as 
\begin{equation}
    \begin{aligned}
        \underset{\delta^{(j)}  \in \bar{\Gamma}^{(j)}}{\max}&\;\delta^{(j)}(\psi_j)  p(H_1|\;\psi^{(j)}), \\
        \text{s.t.} &\; \delta^{(j)}(\psi_j)p(H_0|\psi^{(j)})<\beta^{(j)}.
        \label{prob:multi_stage_human}
    \end{aligned}
\end{equation}
We still inherit the substitutions and characterize the defender's strategies at stage $j$ as the pairs $u^j_1,u^j_0\in \mR^S$. Then we can equivalently express the defender's problem at stage $j$ upon knowing $\delta^{(j)*}$ as follows: 
\begin{equation}
\begin{aligned}
  \underset{u^j_1,u^j_0\in \mR^S}{\max}\;  \sum_{s \in \mc{R}^j_s }{f_1(s)u^j_1(s)}
          \label{prob:multi_stage_defender}
\end{aligned}
\end{equation}
with $R^j_s =\{s:  f_1(s)u^j_1(s)>\tau f_0(s)u^{j}_0(s)\}$.
We now argue that the sequential perfect Bayesian Nash equilibrium (s-PBNE) by applying one-shot deviation principle \cite{fudenberg1998game} into solving \eqref{prob:multi_stage_defender} and \eqref{prob:multi_stage_human}.
\begin{prop}
\label{prop:dynamic_equilibrium_strategy}
    Let $\mc{G}^N$ be the multistage game of finite horizon $N$. Let the assumption \ref{assume:no_change_classical_message} hold. The samples of signals generated during the $j$ stages are denoted as $\{s_t\}_{t\leq j}$. Then we derive the sequential perfect Bayesian Nash equilibrium as the following tuple $\langle u^{j*}_{1}, u^{j*}_{0},\delta^{(j)*}, p\rangle$ as 
    \begin{align}
        u^{j*}_{1}(s) & =\begin{cases}
            0 & \tau^{(j)}_s>1, \\
            1 - \tau^{(j)}_s & \mbox{otherwise}.
        \end{cases}
        \label{defender_strategy_0_time}
        \\
        u^{j*}_{0}(s)& = \begin{cases}
             0 & \tau^{(j)}_s>1, \\
            1  & \mbox{otherwise}.
        \end{cases}
                \label{defender_strategy_1_time}
        \\
        \delta^{(j)*}(\psi_j|\;h^{(j-1)})& =  \begin{cases}
            1&  \prod_{t\leq j-1}{\frac{f_1(s_t) (a^{(t)}_{sk})^2}{f_0(s_t) ((b^{(t)}_{sk})^2 }}> \Big(\frac{1}{\beta^j}-1\Big)\frac{p(H_0)}{p(H_1)}, \\
        0 &  \mbox{otherwise}
        \end{cases}
      \label{attacker_strategy_time}
    \end{align}
\end{prop}
\begin{proof}
    We can derive the equilibrium by backward induction \cite{zamir_game_theory_cambridge}, alternatively solving the optimization problem for every stage $j\in [N]$.
\end{proof}

The equilibrium results in proposition \ref{prop:dynamic_equilibrium_strategy} implies how the defender should change the way of configuring prospect states produced by the generator based human attacker's action history and similarly, how the human attacker adopts her optimal decision threshold based on the history of classical signals received.

\section{Case Study: honeypot detection}
\label{sec:case_study}

In this section, we apply the proposed cyber deception scheme discussed in section \ref{sec:formulation} to implement cyber-psychological techniques to build next-generation honeypots \cite{spitzner2003honeypots_insider_threat} to mitigate inside human attacks. A honeypot is a monitored and regulated decoy disguised as a valuable asset to attract attackers to compromise so as to detect, deflect, and to gather information for cyber attacks in networks. According to \cite{camerer2011behavioral_game_honeypot}, honeypots can help enhance system security in the following ways: to begin with, honeypots squander attacker's resources without giving away valuable information in return; also, honeypots serve as an intrusion detection node, providing warnings for system administrative; last but not least, once compromised, honeypots provide useful information for network administrative to analyze on the attacker.  However, honeypots can also be identified by proactive attackers and become ineffective, especially when they are at fixed locations and isolated from the network system. Attackers can adopt proactive detection techniques, such as those in \cite{hu2022_evasion_game_CDC}, to identify honeypots more accurately and further either implement anti-honeypot techniques \cite{krawetz2004anti_honeypot_technology}. Here inspired by the experiments introduced in \cite{ferguson2021decoy_psychology}, we undermine the attacker's performance in identifying honeypots using cyber-psycholoigical techniques. Specifically, we adopt generators to produce verbal messages to change the perception of attacker's judgment upon the type of the sensors that they receive traffic data from. 

\subsection{The dataset}
To simulate normal traffic and honeypot-related traffic, we select a portion of the KDD cup 1999 dataset \cite{misc_kdd_cup_1999_data_130}, which was generated partially for the 1998 DARPA intrusion detection program. The raw dataset is binary, containing five million connection records from 7 weeks of network traffic. 
There are in total $N= 494021$ pieces of connection records in our selected dataset, of which $396093$ come from honeypot-related traffic. We assume that all attack traffic is attracted by honeypots and all normal traffic is collected by normal sensors since regular users have no reason to access honeypots. Thus we can estimate a prior belief regarding the type of sensors as 
$p(H_1) \approx 0.802,\; p(H_0)  \approx 0.198.$

The signal $s$ of log-in attempts as a feature obeys Bernoulli distribution: $s=0$ means that log-in attempt is successful and $s=1$ means the log-in attempt is a failure. Honeypots and normal sensors respond failure/success to attacker's log-in attempts with different parameters $\tha_1,\tha_0$, i.e. 
  \begin{equation}
\begin{aligned}
    g_1(1)  &= \tha_1,\; g_1(0) =  1-\tha_1,\; \\ g_0(1)  &= \tha_0,\; g_0(0) = 1-\tha_0.
    \label{g_1g_0}
\end{aligned}
\end{equation} 
In our selected dataset, there are $69939$ out of $97928$ successful log-in attempts in normal traffic, while there are $3298$ out of $396093$ successful log-in attempts in honeypot traffic. Thus we select $\tha_0 \approx 0.719, \tha_1 \approx 0.008$.
We adopt the scheme in Figure \ref{fig:human_sensor_system_b} to implement psychological techniques to mislead attackers by creating perceptions of messages based on the traffic data. 
For every message $s$, the generator produces verbal messages leading the human attackers to generate perceptions regarding the type of sensor. 

We denote the number of base states in the mind space to be $K$, indicating how likely the human attacker thinks the defender is a honeypot: upon receiving the classical signal $s$, the state $|s\varphi_1\rangle$ indicates that human thinks the defender is the most likely a honeypot, while the state $|s\varphi_K\rangle$ represents that human thinks the defender is the least likely the honeypot. If the generator wants to keep the human agent uninformed, as described in \cite{ferguson2021decoy_psychology}, he can produce a prospect state as a uniform superposition of all the basis as follows:
    \begin{equation}
        |\Phi(s)\rangle = \frac{1}{\sqrt{K}}\sum_{j=1}^K{|s\varphi_j \rangle }.
    \end{equation}
    
For every realization of message $s$, the generator designs optimal type-dependent prospect state coefficients $a^*,b^*\in \mR^{S\times K}$ in \eqref{util_1}-\eqref{attr_0} via PBNE in proposition \ref{prop:generator_equil_strategy}, indicating perception of the likelihood of honeypot the defender imposes the human agent upon delivering the message $s$. We also analyze the optimal decision rules of human agents under the verbal messages of generators at equilibrium. 

\subsection{Numerical Results}
We select parameters $\beta = 0.4,\zeta=0.2$ and the number of base states in the mind space $K=4$. In Figure \ref{fig:u1s_u0s_prior}, we plot the cyber defender's optimal strategies $u^*_0,u^*_1$ at equilibrium in terms of various choices of $\beta$. We observe that in the classical rejection region (i.e. the space of signals that causes a human attacker to identify that the sensor is a decoy), the generator in the defender system produces perceptions leading to only `rejects' with a certain probability. In Figure \ref{fig:ask_bsk_one_shot} we plot the defender's strategies at equilibrium in terms of the coefficients $a,b$ of the prospect states produced by the generators. The coefficients suggest an optimal way of mixing different weights of psychological minds regarding every classical signal $s$. We observe that when $\beta$ becomes close to $1$, the defender's equilibrium strategies are close to $u_1(0) = 1,u_0(0) = 1$. On the other hand, if $\beta$ is close to $0$, the defender's strategies converge to $u_1 = 0, u_0 = 0$, corresponding to the upper right and lower left corner of the ROC curves (to be described later) characterizing the detection performance. 

\begin{figure}
    \centering
 \includegraphics[width=0.9\linewidth]{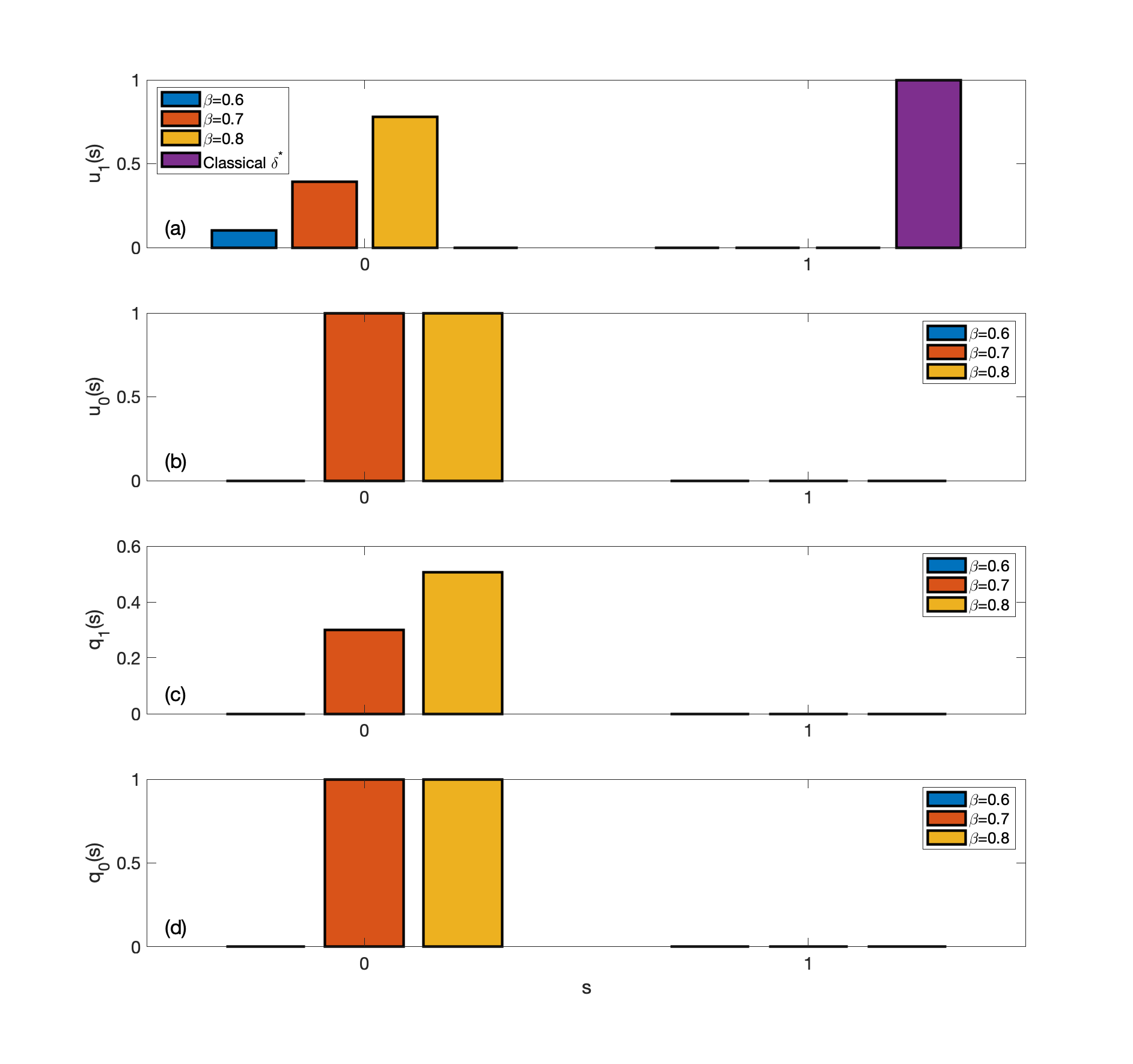}
    \caption{The defender's optimal strategies $u^*_1$ (upper figure) and $u^*_0$ (lower figure) at PBNE in $\mc{G}$ under different choices of $\beta$. We set the calibration parameter $\zeta=0.2$ and the tolerance $\beta=0.4$. The classical signal obeys truncated Gaussian as in \eqref{g_1g_0} with support of length $S=2$. The dimension of mind states $K=4$.}
    \label{fig:u1s_u0s_prior}
\end{figure}

\begin{figure}
    \centering
    \includegraphics[width=0.6\linewidth]{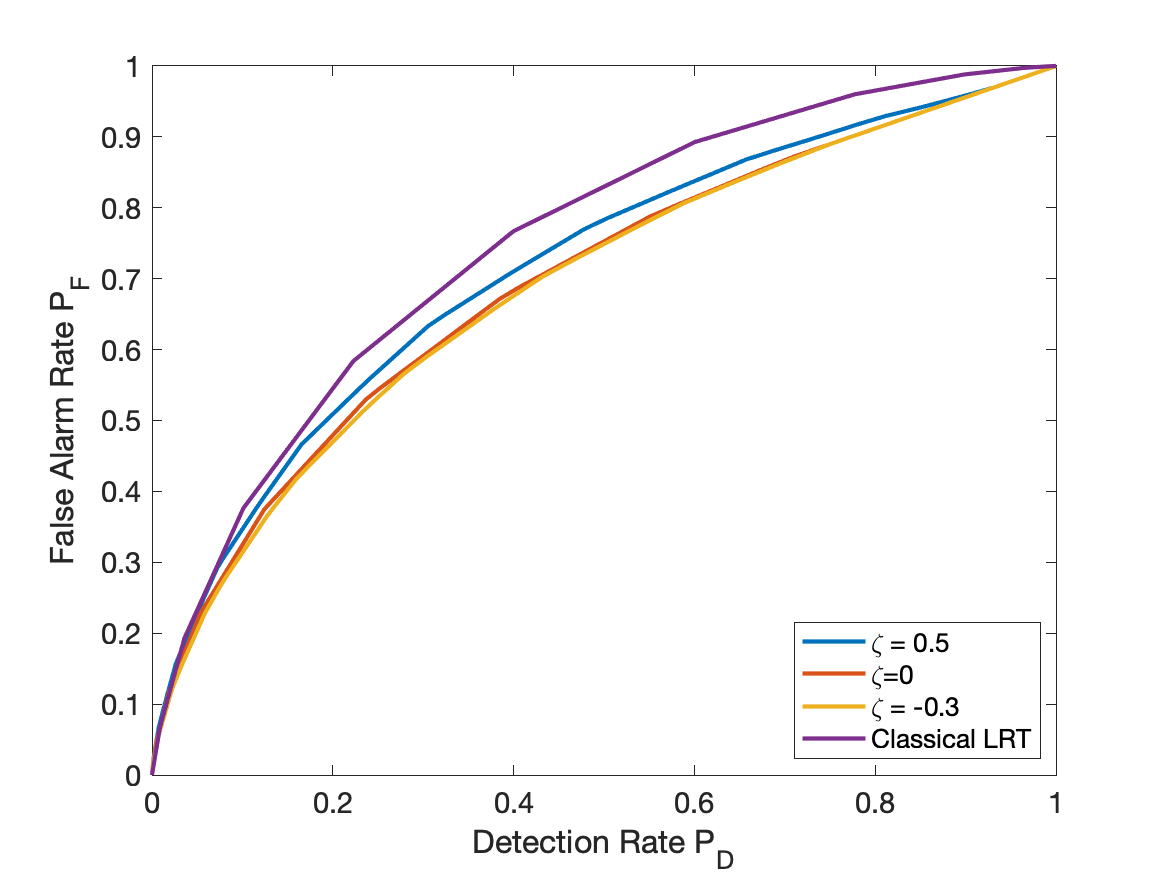}
    \caption{ROC curves human agent's detection performance $\zeta$. We choose distributions of classical signals under each state as $g_1,g_0$ given in \eqref{g_1g_0}}
    \label{fig:ROC_curve}
\end{figure}

\begin{figure}
    \centering
\includegraphics[width=0.8\linewidth]{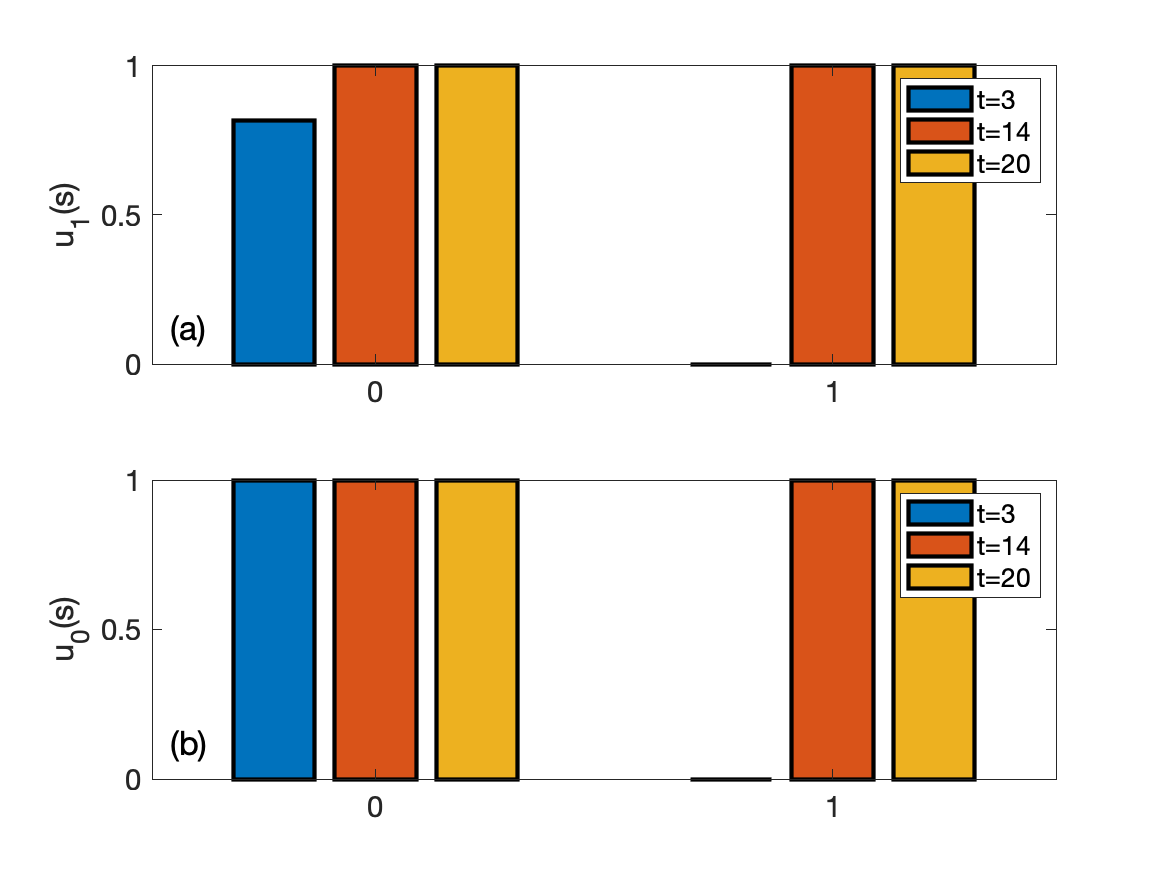} \\
\includegraphics[width=0.8\linewidth, height = 0.5\linewidth]{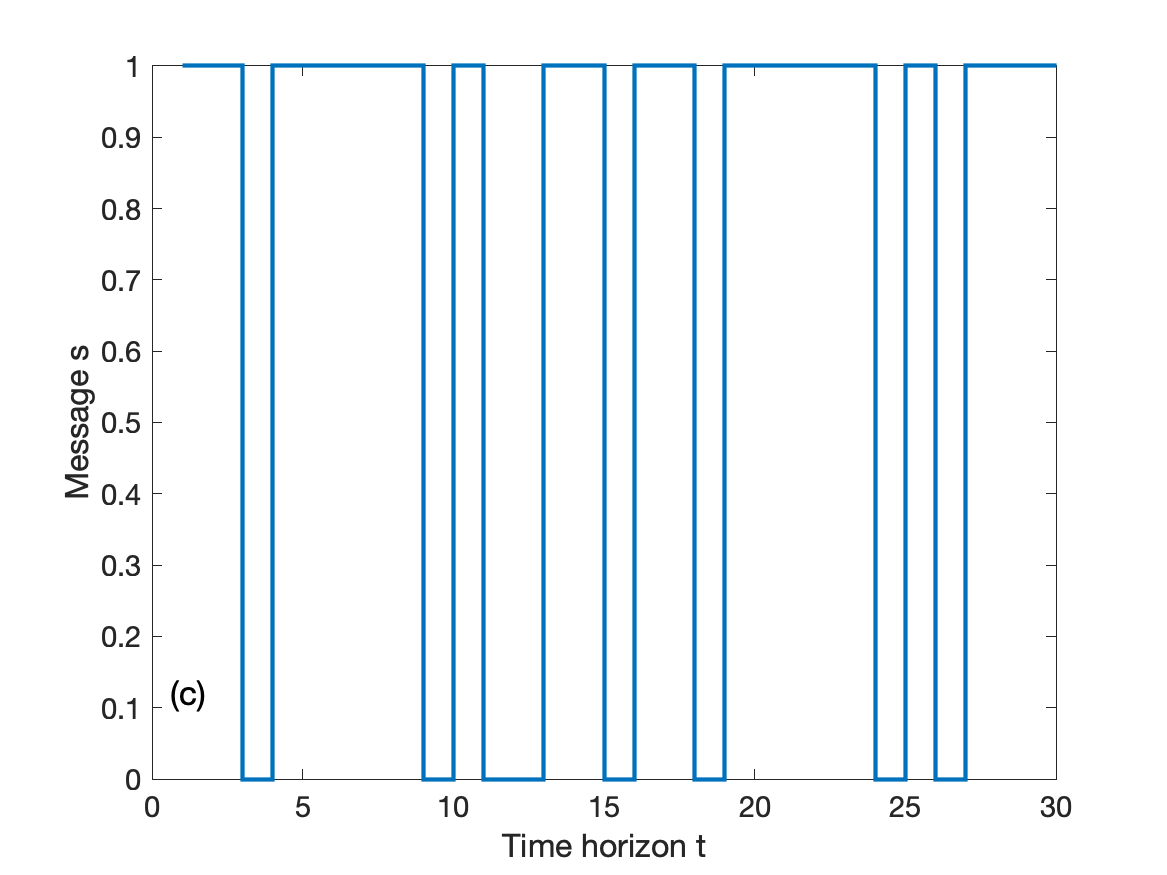} 
    \caption{The defender's optimal strategies $u_1,u_0$ at sequential PBNE in multi-stage game $\mc{G}^N$: a)b) the attacker's optimal strategies $u_1(s;t),u_0(s;t)$ at equilibrium; c) the history of realization of classical signals that the defender receives from Nature. We adopt the setting in \eqref{g_1g_0} and fix $\zeta= 0.5$.}
    \label{fig:u0s_u1s_sample_s}
\end{figure}

\begin{figure}
    \begin{subfigure}
       \centering
    \begin{tabular}{c c}
\includegraphics[width=0.45\linewidth]{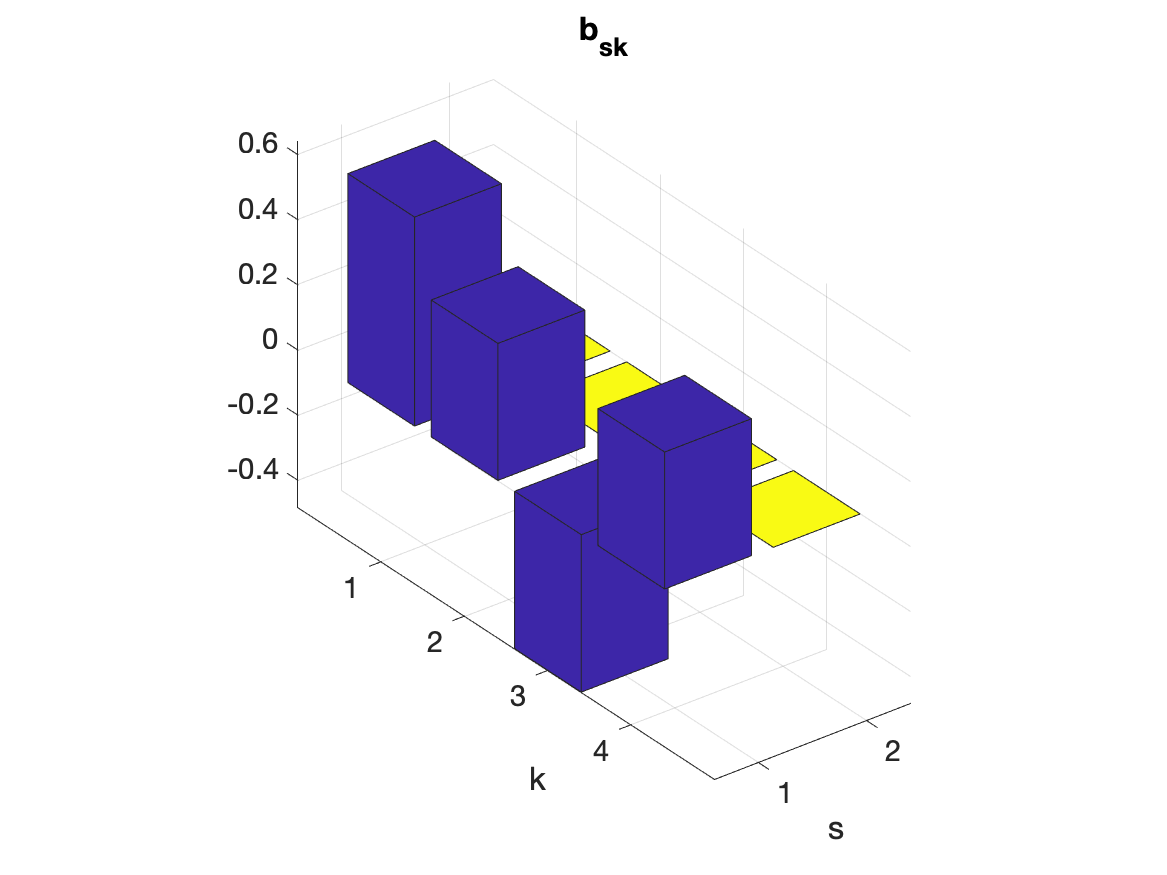} & \includegraphics[width=0.45\linewidth]{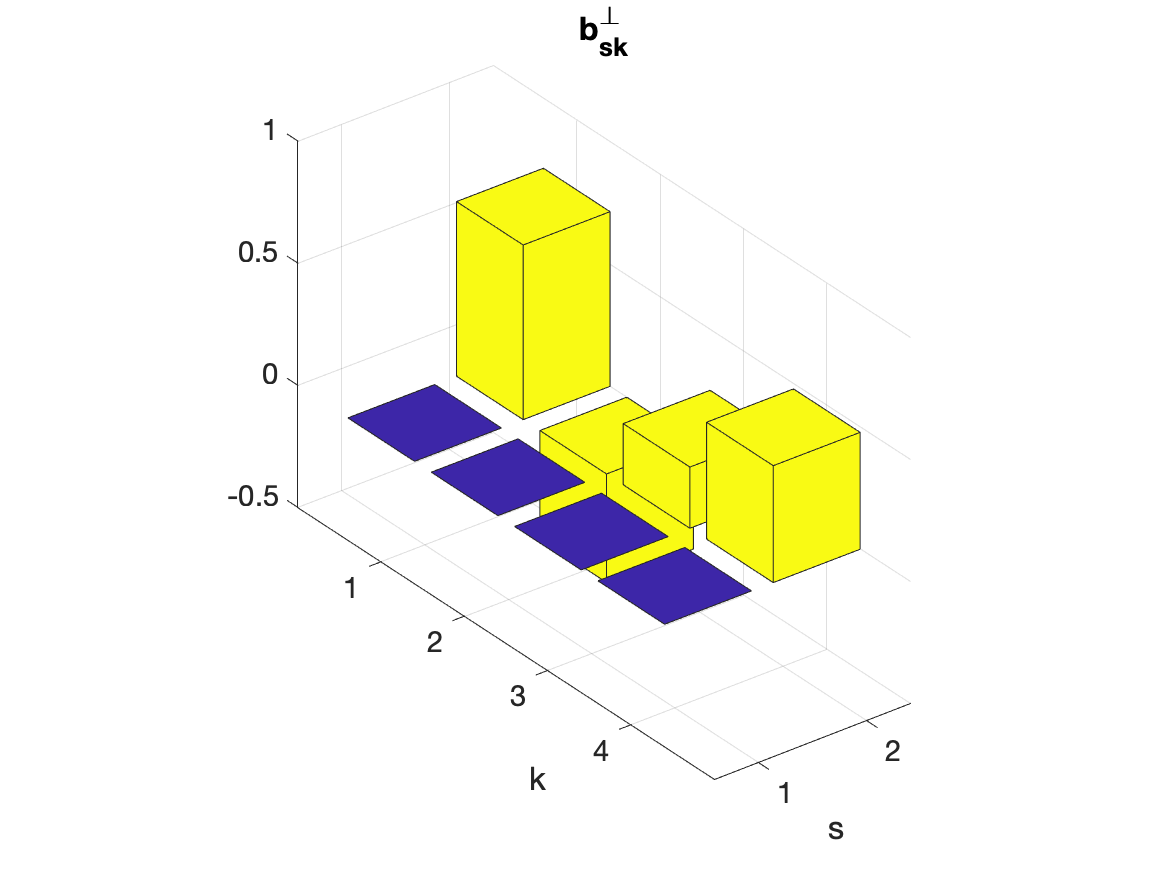} \\
\includegraphics[width=0.45\linewidth]{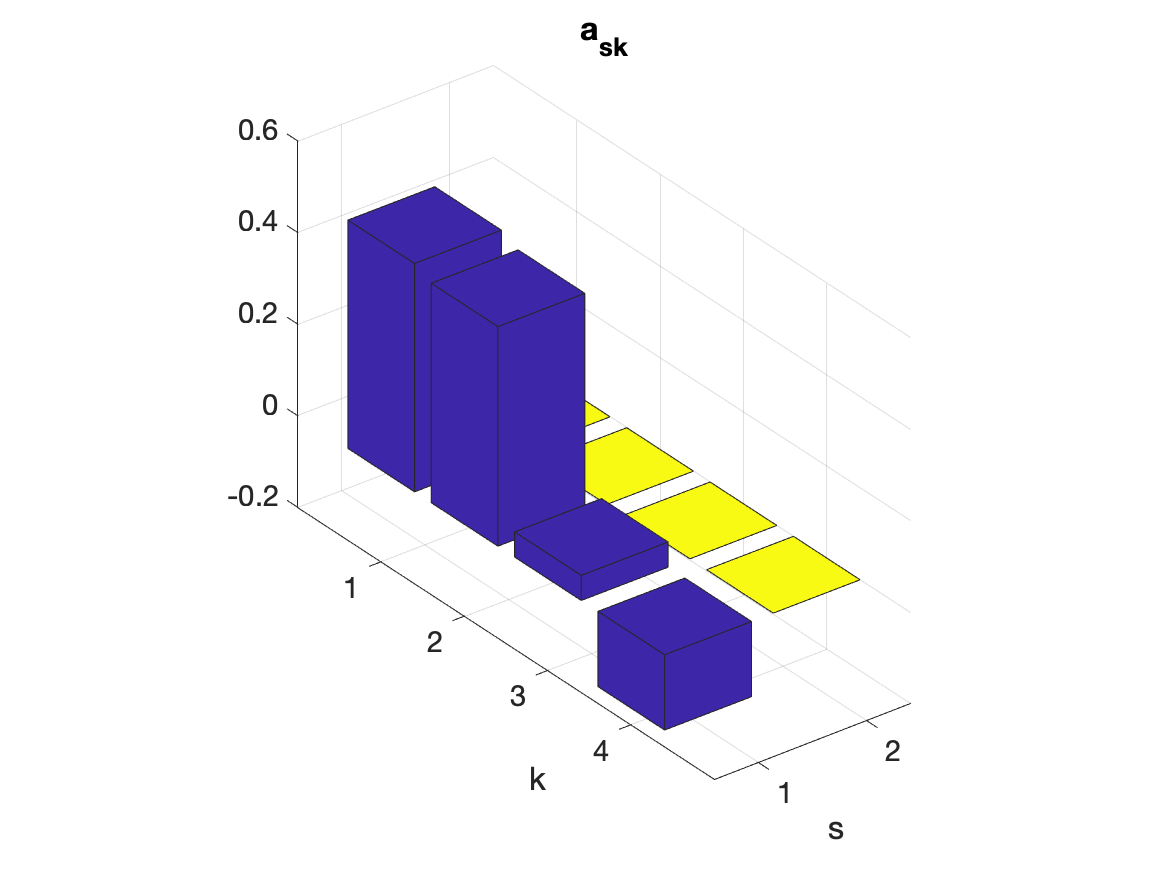} & \includegraphics[width=0.45\linewidth]{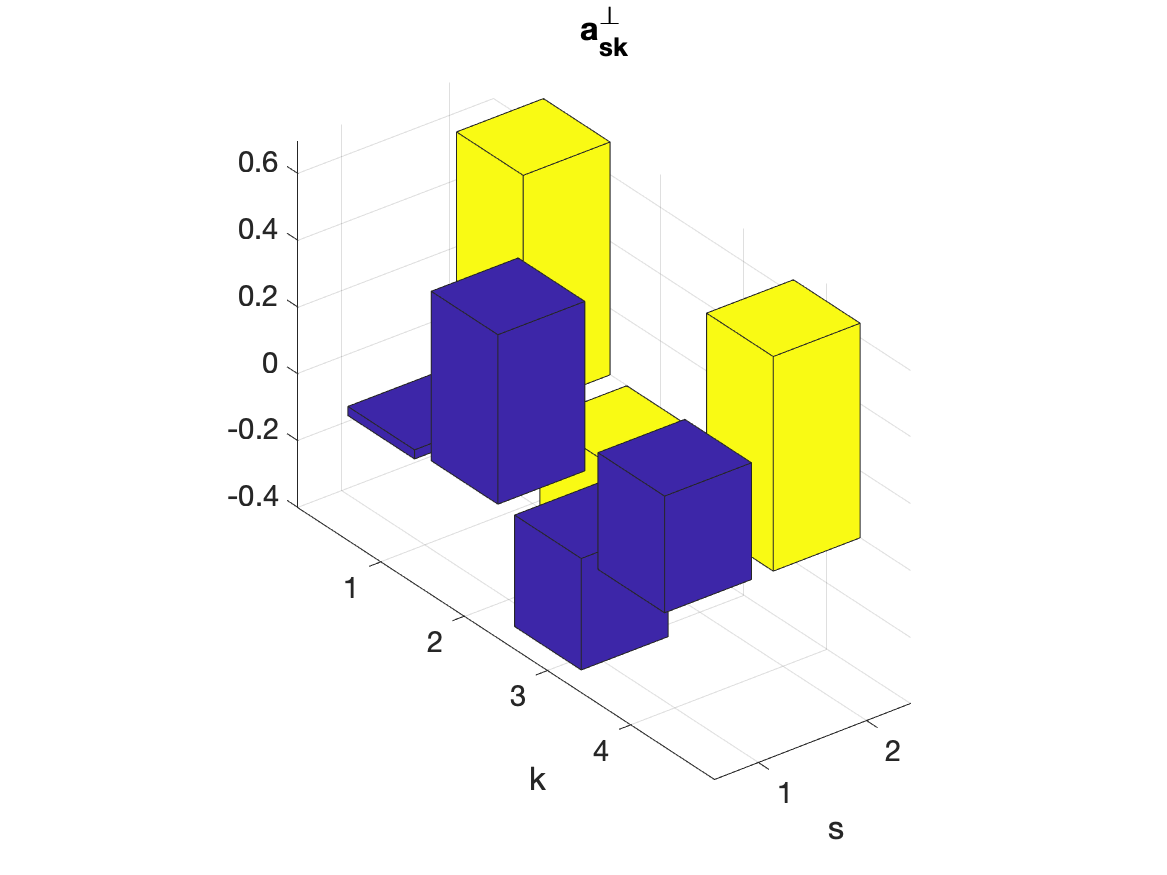}
    \end{tabular}
     \end{subfigure}
     \caption{The history of coefficients for prospect states $a_{sk},b_{sk},a^{\perp}_{sk},b^{\perp}_{sk}$ for the static travesty game $\mc{G}$. The classical message $s$ obeys the distribution $g_0,g_1$ under hypothesis $H_0,H_1$ respectively. We set $S = 2$ and $K=4$ and the prior belief $p(H_1) = 0.802$.}
    \label{fig:ask_bsk_one_shot}
\end{figure}

In Figure \ref{fig:ROC_curve},  we plot the receiver-operational-characteristic (ROC) curve. We observe that depending on different calibration parameters $\zeta$, the human agent's detection performances vary, but in general are all worse than fully rational counterparts. In particular, higher $\zeta$ leads to better detection performance, as higher $\zeta$ the quantum inference strengthens the probability of correct identification when the decoy sensor is connected to the human agent. 

\subsection{Multi-stage analysis}
To have a better understanding of the 
In Figure \ref{fig:u0s_u1s_sample_s}, we plot the evolution of defender's optimal strategies $\{u^{t*}_1,u^{t*}_0\}_{t\in[N]}$ (that is, optimal type-dependent utility factors) through time at equilibrium as introduced in \eqref{defender_strategy_1_time}\eqref{defender_strategy_0_time}. We select the time horizon $N=30$ and fix the prior belief $p(H_1),p(H_0)$. We observe that the defender's stage equilibrium strategies converge to a pooling strategy $u_0=u_1 = 1$, suggesting that the attacker can make the prospect states totally uninformative to the human agent by designing false perceptions upon the signals.

\section{Conclusion}
\label{sec:conclusion}
In this work, we have proposed the game of travesty (TG) to design a novel defensive deception to combat the proactive detection of decoys from insider human attackers. The defensive deception system is 
a signaling game where the defender consists of a sensor or a decoy cascaded by a generator, which converts classical signals into prospect states to manipulate the perception of messages into human attackers. 
We have analyzed the behaviors of the inside human attacker as well as the defender by computing the perfect Bayesian Nash equilibrium. Furthermore, we analyze the human attacker's performance of detecting decoys at equilibrium and compare it with the ones without manipulation of perceptions of classical signals. We have illustrated via ROC curves that the insider human attacker performs worse than the ones with full rationality, giving the defender more room to evade detection when she implements decoys in the network.   

\bibliographystyle{IEEEbib}
\bibliography{thesis}

\section{Appendix}
\subsection{Proof of proposition in \ref{prop:equiv_optimization_defender}}
\label{appd:equiv_optimization_defender}
\begin{proof}
    We can prove its contrapositive. Let $a^*,b^*$ be an optimal solution of the problem \eqref{problem:P1}. Denote $u^*_1(s) = \tr(a^*_s a^{*T}_s), u^*_0(s) =\tr(b^*_s b^{*T}_s) $.  If there exists $\tilde{u}_1,\tilde{u}_0\in [0,1]^S$ leading to a smaller objective, i.e. $u^*_1,u^*_0\in [0,1]^S$ are the optimal solution of \eqref{problem:P1}, then by assumption \ref{assume:existence_a_b_coefficient}, we can find $\tilde{a}\in \mR^{S\times K},\tilde{b}\in \mR^{S\times K}$ such that 
    $\tilde{u}_1 = \tr(\tilde{a}\tilde{a}^T),\;\tilde{u}_0 = \tr(\tilde{b}\tilde{b}^T)$. In this way, we derive 
    \begin{equation}
         J_D(\tilde{a},\tilde{b}) = F(\tilde{u}_1,\tilde{u}_0)< F(u^*_1,u^*_0) = J_D({a}^*,{b}^*).
    \end{equation}
   The last inequality reveals that $a^*,b^*$ is not the minimizer of $F$ in the optimization problem, contradicting our assumption. So $u^*_1,u^*_0$ must be the optimal solution of the optimization problem \eqref{problem:P2}.
\end{proof} 

\subsection{Proof of the proposition in \ref{prop:human_agent_optimal}
}
\begin{proof}
    We could discuss by cases. If $p(H_0|\;\ketsk)<\beta$, picking $\delta(\ketsk) = 1$ or $0$ both meet the inequality constraint, but obviously $\delta(\;\ketsk) = 1$ leads to a higher objective function value. If on the other hand $p(H_0|\;\ketsk)\geq \beta$, we must pick $\delta(\ketsk) = 0$ so as not to violate the constraint. Now we rewrite the inequality constraints by substituting the posterior belief with its expression in \eqref{posterior_belief} as 
    \begin{equation}
        p(H_0|\;\ketsk)<\beta \Leftrightarrow \frac{\brask \rho_1 \ketsk }{\brask \rho_0 \ketsk}>\frac{1}{\beta}\frac{p(H_1)f_1(s)a^2_{sk}}{p(H_0)f_0(s))b^2_{sk}} 
    \end{equation}
    Denoting $\tau_{sk} =\frac{p(H_1)}{p(H_0)} \frac{f_1(s)a^2_{sk}}{f_0(s)b^2_{sk}}$, we can rewrite condition for optimal action $\delta(\ketsk)$ to be 1 as 
    \begin{equation}
        \brask \rho_1 - \tau_{sk}\rho_0 \ketsk > 0.
        \label{ineq: optimal_decision_rule}
    \end{equation}
\end{proof}

\subsection{Proof of proposition in \ref{prop:generator_equil_strategy}
}
\label{appd:sol_signaling_lam_0}
\begin{proof}
    First of all, we notice the objective $F$ is separable in terms of every component $s\in S$. Thus it suffices to minimize $f_1(s)u_1(s)$ for every $s$. Since $f_1,f_0$ are fixed, we only need to aim at minimizing $u_1(s)$ for every $s$ taking into the constraint characterizing the set $\mc{R}_s$. In the meantime, we remember $0\leq u_1(s)\leq 1$ for every $s$. Thus we first try if there is a chance that $u_1(s)=0$ if not we can then pick a positive number for $u_1(s)$ that is as small as possible.
    
    Denoting $\tau_s = \frac{f_0(s)p(H_1)}{f_1(s)p(H_0)}(\frac{1}{\beta}-1)$, we can discuss by cases:
    \begin{enumerate}
        \item When $\tau_s>1$, we can pick $u_1(s) = u_0(s) =0$ and thus $\mc{R}_s = \varnothing$. On the other hand $1 = 1-u_1(s)<\tau_s( 1-u_0(s)) = \tau_s$ also holds.  
        \item When $0<\tau_s<1$, the previous discussion does not directly carry over as 
        $1-u_1(s)\leq \tau_s$ therefore $u_1(s)\geq 1 - \tau_s$. Thus the optimal solution $u^*_1(s)\geq 1 - \tau_s$. Correspondingly, $1-u^*_0(s)=1$ so $u^*_0(s) = 0$. Thus we conclude the proof.
    \end{enumerate}
\end{proof}
\subsection{Proof of proposition \ref{prop:PD_PD_bar}}
\label{appd:PD_PD_bar}
\begin{proof}
    We refer to the expression \eqref{PD_PF_classical} and \eqref{PD_PF_quantum} and discuss via different regimes of $\tau_s$ as follows:
    \begin{equation}
        P_D(\tau) = \sum_{s}{f_1(s)(u_1(s)+q_1(s))}
        \end{equation}
        \begin{enumerate}
            \item  When $\tau_s<1$, or equivalently $\frac{f_1(s)}{f_0(s)}<\frac{p(H_1)}{p(H_0)}(\frac{1}{\beta}-1)=\tau$, the signal $s$ lies in the classical region of acceptance. The probability that $s$ leads to a rejection decision is zero both in classical decisions and in our proposed framework.
            \item   When $\tau_s>1$, or equivalently $\frac{f_1(s)}{f_0(s)}>\frac{p(H_1)}{p(H_0)}(\frac{1}{\beta}-1)=\tau$, the signal $s$ leads to a rejection decision in classical setting. In our framework, it leads to the defender's equilibrium strategy $u^*_1(s) = 1-\tau_s$. Therefore we have 
            \begin{equation}
            \begin{aligned}
                {P}_D(\tau) &= \sum_{s}{f_1(s)(u_1(s)+q_1(s))} \\
                &= \sum_{s,0<\tau_s<1}{f_1(s)(1-\frac{1}{\tau_s})(1+\zeta)} \\
                &<\sum_{s,0<\tau_s<1}{f_1(s)(1+\zeta)}=: \bar{P}_D(\tau)(1+\zeta).
            \end{aligned}
            \end{equation}
        \end{enumerate}
        That concludes the proof. 
\end{proof}

\end{document}